\documentclass[aps,pra,twocolumn,superscriptaddress,nofootinbib]{revtex4-1}  
\usepackage{graphicx}  
\usepackage{dcolumn}   
\usepackage{bm}        
\usepackage{amssymb}   
\usepackage{xspace}

\usepackage{hhline}
\usepackage{scalefnt}
\usepackage{amsmath}
\usepackage{color}
\usepackage{xr}
\usepackage{cases}
\usepackage{enumerate}
\usepackage{tikz}
\usetikzlibrary{matrix,arrows}
\usepackage{hyperref}
\usepackage{enumitem}
\usepackage{relsize}
\def\I {\mathrm{i}}

\usepackage{slashbox}
\input{qcircuit.sty}

\newcommand{\ket}[1]{\left| #1 \right>} 
\newcommand{\bra}[1]{\left< #1 \right|} 
\newcommand{\braket}[2]{\left\langle#1|#2\right\rangle}
\newcommand{\ketbra}[2]{|#1\rangle\!\langle#2|}
\newcommand{\proj}[1]{|#1\rangle\!\langle#1|}
\newcommand{\id}{\leavevmode\hbox{\small1\normalsize\kern-.33em1}}
\newcommand{\tr}{\mathrm{tr}}

\newcommand{\cnot}{{\sc C-not}}
\newcommand{\cnots}{\cnot s}

\newcommand{\cE}{\mathcal{E}}

\newcommand{\n}{N}
\newcommand{\m}{M}

\newcommand{\rgate}{\text{{\sc R}}\xspace}
\newcommand{\xx}{\text{{\sc XX}}\xspace}

\newcommand{\com}[1]{{\sf #1}}
\newcommand{\UQC}{\textit{UniversalQCompiler}}
\newcommand{\QI}{\textit{QI}}

\newcommand\restr[2]{{
  \left.\kern-\nulldelimiterspace 
  #1 
  \vphantom{\big|} 
  \right|_{#2} 
  }}

\AtBeginDocument{%
    \newwrite\bibnotes
    \def\bibnotesext{Notes.bib}
    \immediate\openout\bibnotes=\jobname\bibnotesext
    \immediate\write\bibnotes{@CONTROL{REVTEX41Control}}
    \immediate\write\bibnotes{@CONTROL{%
    apsrev41Control,author="08",editor="1",pages="1",title="0",year="0"}}
     \if@filesw
     \immediate\write\@auxout{\string\citation{apsrev41Control}}%
    \fi
  }%
\usepackage{xcolor}
\definecolor{mylinkcolor}{rgb}{0,0,0.8} 
\hypersetup{unicode=true,%
  bookmarksnumbered=false,bookmarksopen=false,bookmarksopenlevel=1, %
  breaklinks=true,pdfborder={0 0 0},colorlinks=true}%
\hypersetup{%
  anchorcolor=mylinkcolor,citecolor=mylinkcolor, %
  filecolor=mylinkcolor,linkcolor=mylinkcolor, %
  menucolor=mylinkcolor,runcolor=mylinkcolor, %
  urlcolor=mylinkcolor}

\newtheorem{thm}{Theorem}
\newtheorem{lem}[thm]{Lemma}

\newenvironment{proof}[1][Proof]{\noindent\textbf{#1.} }{\ \rule{0.5em}{0.5em}}

\begin{document}

\title{Introduction to \UQC{}}

\author{Raban Iten}
\email{itenr@ethz.ch} 
\affiliation{Institute for Theoretical Physics, ETH Z\"urich, 8093 Z\"urich, Switzerland}
\author{Oliver Reardon-Smith}
\affiliation{Department of Mathematics, University of York, YO10 5DD, UK}
\author{Emanuel Malvetti}
\affiliation{Department of Chemistry, TUM, Lichtenbergstra{\ss}e 4, 85747 Garching, Germany}
\author{Luca Mondada}
\affiliation{Institute for Theoretical Physics, ETH Z\"urich, 8093 Z\"urich, Switzerland}
\author{Gabrielle Pauvert}
\affiliation{Department of Mathematics, University of York, YO10 5DD, UK}
\author{Ethan Redmond}
\affiliation{Department of Mathematics, University of York, YO10 5DD, UK}
\author{Ravjot Singh Kohli}
\affiliation{Department of Mathematics, University of York, YO10 5DD, UK}
\author{Roger Colbeck}
\email{roger.colbeck@york.ac.uk}
\affiliation{Department of Mathematics, University of York, YO10 5DD, UK}

\date{$29^{\mathrm{th}}$ March 2021}

\begin{abstract}
  We introduce an open source software package \UQC{} written in
  Mathematica that allows the decomposition of arbitrary quantum
  operations into a sequence of single-qubit rotations (with arbitrary
  rotation angles) and controlled-NOT (\cnot{}) gates.  Together with
  the existing package \QI{}, this allows quantum information
  protocols to be analysed and then compiled to quantum circuits. Our
  decompositions are based on \href{\doibase 10.1103/PhysRevA.93.032318}{Phys.\ Rev.\ A {\bf 93}, 032318 (2016)},
  and hence, for generic operations, they are near optimal in terms of
  the number of gates required. \UQC{} allows the compilation of any
  isometry (in particular, it can be used for unitaries and state
  preparation), quantum channel, positive-operator valued measure
  (POVM) or quantum instrument, although the run time becomes
  prohibitive for large numbers of qubits. The resulting circuits can
  be displayed graphically within Mathematica or exported to
  \LaTeX. We also provide functionality to translate the circuits to
  OpenQASM, the quantum assembly language used, for instance, by the
  IBM Q Experience.
\end{abstract}

\maketitle

\section{Introduction} \label{sec:intro}
A universal quantum computer should be able to perform arbitrary
computations on a quantum system. It is common to break down a given
computation into a sequence of elementary gates, each of which can be
implemented with low cost on an experimental architecture. However,
given an abstract representation of the desired computation, such as a
unitary matrix, it is in general difficult and time consuming to find
a low-cost circuit implementing it. Here, we introduce an open source
Mathematica package,
\href{http://www-users.york.ac.uk/~rc973/UniversalQCompiler.html}{\UQC{}}\footnote{See
  our webpage for a reference to the github repository and the
  documentation:
  \url{http://www-users.york.ac.uk/~rc973/UniversalQCompiler.html}.}
that allows for automation of the compiling process on a small number
of qubits. The package requires an existing Mathematica package
\QI\footnote{\url{https://github.com/rogercolbeck/QI}}, which can
easily handle common computations in quantum information theory, such
as partial traces over various qubits or the Schmidt
decomposition. Since the code is provided for Mathematica, our
packages are well adapted for analytic calculations and can be used
alongside the library of mathematical tools provided by
Mathematica. Together, these constitute a powerful set of tools for
analysing protocols in quantum information theory and then compiling
the computations into circuits that can finally be run on a
experimental architecture, such as IBM Q Experience (see
Figure~\ref{fig:overview} for an overview). \UQC{} focuses on the
compilation process, and performs a few basic simplifications on the
resulting quantum circuit. Hence, one might want to put the gate sequences
obtained from \UQC{} into either a source-to-source compiler or a transpiler 
(see for example~\cite{maslov,maslov_cont_par,xzCalculus}) in order to optimize the gate count of 
the circuits further or to map them to a different hardware, which may have 
restrictions on the qubit-connectivity~\cite{qubit_mapping1,qubit_mapping2,qubit_mapping3,qubit_mapping4,qubit_mapping5}.

\begin{figure*}[!t] 
\includegraphics[width=1\textwidth]{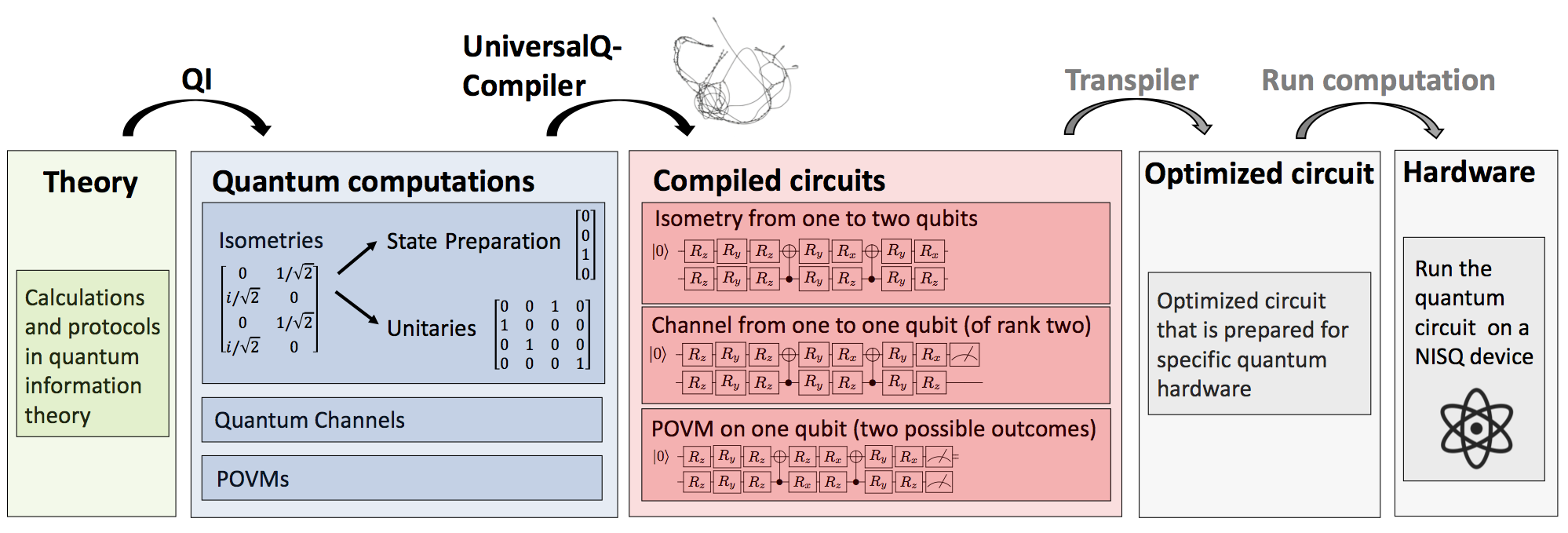}
\caption{Overview over the use of \UQC{}. The Mathematica package
  \QI{} can be used to do common computations in quantum information
  and manipulate quantum operations, such as unitary matrices. Given
  an abstract representation of a quantum computation, \UQC{} takes it
  as an input and outputs a quantum circuit implementing the
  computation. Thereby, we distinguish the following classes of
  operations: isometries (including unitaries and state preparation as
  special cases), quantum channels and POVMs. The picture depicts
  example-circuits for each class of operations on two qubits. The
  circuit could then be further optimized and prepared for a specific
  quantum hardware architecture by an (external) transpiler. To
  simplify interfacing with a transpiler or a circuit optimizer (for
  instance the \href{https://github.com/Quantomatic}{\textit{PyZX}}
  quantum circuit optimizer~\cite{xzCalculus}), we provide a python
  script (based on \textit{ProjectQ}~\cite{ProjectQ1,ProjectQ2}) to translate the Mathematica outputs to the quantum assembly
  language OpenQASM.}
  \label{fig:overview}
\end{figure*}

The package \UQC{} provides code for all the decompositions described
in~\cite{Iso}, which are near optimal in the required number of gates
for generic computations in the quantum circuit model (in fact, the
achieved \cnot{} counts differ by a constant factor of about two from
a theoretical lower bound given in~\cite{Iso}). Note that our
decompositions may not lead to optimal gate counts for computations of
a special form lying in a set of measure zero (see~\cite{Iso} for the
details), as for example for a unitary that corresponds to the circuit
performed for Shor's algorithm~\cite{shor}.  Hence, to optimize the
gate counts when decomposing operations of certain special forms, such
as diagonal gates, multi-controlled single-qubit gates and
uniformly-controlled gates, we provide separate commands. In addition,
we provide methods for analyzing, simplifying and manipulating gate
sequences.  Outputs are given in a bespoke gate list format, and can
be exported as graphics, or to \LaTeX{} using the format of
Q-circuit~\cite{qcirc}.

\UQC{} is intended to be an academic software library that
focuses on simplicity and adaptability of the code and it was not our
focus to optimize the (classical) run time of the decomposition
methods (the theoretical decompositions mainly focused on minimizing the 
\cnot{} count). A detailed documentation as well as an example notebook are
published together with our code and should help the user to get
started quickly. The aim of this paper is to give an overview over the
package \UQC{} and to provide some theoretical background about the
decomposition methods that it uses.  A separate manual is provided
with the package that provides more details.

We work with the universal gate library consisting of arbitrary
single-qubit rotations and \cnot{} gates (we also explain how to
convert gate sequences from this universal set to another that
comprises single-qubit rotations and M{\o}lmer-S{\o}rensen gates (see
Appendix~\ref{app:tranform_gate_library}), which are common on
experimental architectures with trapped ions). \UQC{} decomposes
different classes of quantum operations into sequences of these
elementary gates keeping the required number of gates as small as possible.

In Section~\ref{sec:UGL}, we define the elementary gates we are working with, i.e., the single-qubit rotations and the \cnot{} gate.

In Section~\ref{sec:isometries}, we describe how to use \UQC{}  to decompose arbitrary isometries from $m$ to $n\geq m$ qubits describing the most general evolution that a closed quantum system can undergo. Mathematically, an isometry from $m$ to $n$ qubits is an inner-product preserving transformation that maps from a Hilbert space of dimension $2^m$ to one of dimension $2^n$. Physically, such an isometry can be thought of as the introduction of $n-m$ ancilla qubits in a fixed state (conventionally $\ket{0}$) followed by a general $n$-qubit unitary on the $m$ input qubits and ancilla qubits. Unitaries and state preparation on $n$ qubits are two important special cases of isometries from $m$ to $n$ qubits, where $m=n$ and $m=0$, repsectively.

In Section~\ref{sec:channels}, we consider the decomposition of
quantum channels from $m$ to $n$ qubits (no longer restricting to
$m\leq n$). A quantum channel describes the most general evolution an
open quantum system (i.e., a quantum system that may interact with its
environment) can undergo. Mathematically, a quantum channel is a
completely positive trace-preserving map from the space of density
operators on $m$ qubits to the space of density operators on $n$
qubits. \UQC{} takes a mathematical description of such a quantum
channel (which can be supplied in Kraus representation or as a Choi
state) and returns a gate sequence that
implements the channel (in general after tracing out some qubits at the end of
the circuit). The decomposition is nearly optimal for generic channels
working in the quantum circuit model~\cite{Iso}. However, working in
more general models would allow further reducions in the number of
gates~\cite{channel}. We plan to implement code for the decompositions
described in~\cite{channel} in the future. For an overview of possible
applications of implementing channels,
see~\cite{channel_applications}.

In Section~\ref{sec:POVM}, we describe how to implement arbitrary
POVMs on $m$ qubits describing the most general measurements that can
be performed on a quantum system. Similarly to the case of channels,
working in generalized models can reduce the gate count
further~\cite{channel}, and we plan to implement these in a future
version. See also~\cite{Leo} for an application of \UQC{} for
synthesis of POVMs.

In Section~\ref{sec:inst}, we extend from POVMs to quantum instruments.  These can be thought of as the most general type of quantum measurement where we care about the post-measurement state (in contrast to a POVM where we only care about the distribution over the classical outcomes). Our decompositions for these are based on those used for channels, and again could be improved using additional methods from~\cite{channel}.

In Section~\ref{sec:simplifying}, we describe some simple rules that
can be used to simplify circuits and that are implemented within
\UQC{}.

Finally, in Section~\ref{sec:QASM}, we explain how to automatically
translate our circuits to the open quantum assembly language
(OpenQASM)~\cite{QASM}, which allows our package to interface with
other quantum software packages.

\begin{table*}[!t] 
\renewcommand{\arraystretch}{1.5}
\caption{Overview of the asymptotic number of \cnot{} gates and the classical run time required to decompose $m$ to
  $n$ isometries 
  using different decomposition schemes. Abbreviations used: $^a$Column-by-column
  decomposition of an isometry; $^b$Decomposition of an isometry using
  the Quantum Shannon Decomposition; $^c$State preparation.}
\centering
\begin{ruledtabular}
\begin{tabular}{lllll}
Method & \cnot{} count for a generic $m$ to $n$ isometry & Classical run time & References \\ \hline
CCD$^a$ &$2^{m+n}-\frac{1}{24}2^n+  \mathcal O\left(n^2 \right)2^m$&$\mathcal{O}(n2^{2m+n})$& \cite{Iso}\\
QSD$^b$ & $\frac{23}{144}\left(4^m+2\cdot4^n\right)+  \mathcal O\left(m\right)$&$\mathcal{O}(2^{3n})$& \cite{2},\cite{Iso}\\
Knill & $\frac{23}{24}(2^{m+n}+2^n)+\mathcal O\left(n^2 \right)2^m$ \text{ if $n$ is even} &$\mathcal{O}(2^{3n})$& \cite{Knill},\cite{Iso}\\
		   &$ \frac{115}{96}(2^{m+n}+2^n)+\mathcal O\left(n^2 \right)2^m$ \text{ if $n$ is odd}&$\mathcal{O}(2^{3n})$& \cite{Knill},\cite{Iso}\\
SP$^c$ & $\frac{23}{24}2^n$ [here $m=0$]&$\mathcal{O}(2^{\frac{3n}{2}})$& \cite{3},\cite{Iso}\\  [0.04cm]								   
\end{tabular}
\end{ruledtabular}
\label{tab:iso}
\end{table*}

\section{Universal gate library} \label{sec:UGL}

Our gate library consists of arbitrary single-qubit rotations and \cnot{} gates. This set of gates is known to be universal~\cite{5}, i.e., any quantum computation can be decomposed into a sequence of gates in this set. We use the following convention for rotation gates 

\begin{align}
R_x(\theta)&=\begin{pmatrix}
\cos(\theta/2) & i\sin(\theta/2)\\
i\sin(\theta/2) & \cos(\theta/2) 
\end{pmatrix},\, \\
 R_y(\theta)&=\begin{pmatrix}
                                \cos(\theta/2) & \sin(\theta/2)\\
                                -\sin(\theta/2) & \cos(\theta/2) 
                            \end{pmatrix}, \, \\
                            R_z(\theta)&=\begin{pmatrix}
                                e^{i \theta/2} & 0\\
                                0 & e^{-i \theta/2} 
                            \end{pmatrix}. \
                            \end{align}
                             Note that in~\cite{Iso}, we used the convention $R_x'(\theta)=R_x(-\theta)$, $R_y'(\theta)=R_y(-\theta)$ and $R_z'(\theta)=R_z(-\theta)$.
                            In addition, we use the following two-qubit gate
  \begin{align}                          
  \textnormal{\cnot}=\begin{pmatrix}
1 & 0& 0& 0\\
 0& 1& 0& 0\\
 0& 0& 0& 1\\
 0& 0& 1& 0
\end{pmatrix} \, .
\end{align}

In Appendix~\ref{app:tranform_gate_library}, we explain how to convert
gate sequences from this universal set to the one that comprises
single-qubit rotations and M{\o}lmer-S{\o}rensen gates without
increasing the number of two-qubit gates.

Note that, when displaying circuits, we use
$\Qcircuit @C=0.2em @R=.2em {&\qw&\qw&\meter&\cw&\cw}$ to represent
measurement in the computational basis where the classical outcome is
retained, and $\Qcircuit @C=0.2em @R=.2em {&\qw\qw&\meter}$ to
represent the qubit being traced out (equivalent to measuring and
forgetting the outcome).

\section{Compilation of isometries}\label{sec:isometries}

\UQC{} provides three different decompositions for generic isometries from $m$ to $n$ qubits given as a $2^n \times 2^m$ complex matrix $V$ satisfying $V^\dagger V=\id$ (and an additional decomposition that only works for $m=0$, i.e., for state preparation). For an overview of the gate counts and the running time complexity of the different methods, see Table~\ref{tab:iso}. For a comparison with a theoretical lower bound on the number of \cnot{} gates, see Table~I in~\cite{Iso}. In the following, first we give some information about the different decomposition schemes. Then we explain the methods  \com{DecIsometry} and \com{DecIsometryGeneric}, which choose the optimal decomposition scheme automatically. An example of
an output circuit created by  \com{DecIsometry} is given in Figure~\ref{fig:iso}.\\

\begin{figure*}[!t]
\[
\Qcircuit @C=0.2em @R=.2em {
\lstick{\left| 0 \right>}& \qw  & \gate{R_z} & \gate{R_y} & \gate{R_z} & \targ & \gate{R_y} & \gate{R_x} & \ctrl{1} & \qw  & \qw  & \qw  & \qw  & \qw  & \qw  & \qw  & \qw  & \ctrl{2} & \gate{R_y} & \qw  & \qw  & \qw  & \qw  & \ctrl{2} & \gate{R_y} & \gate{R_z} & \qw  & \qw  & \qw  & \qw  & \qw  & \qw  & \qw  \\
\lstick{\left| 0 \right>}& \qw  & \gate{R_z} & \gate{R_y} & \gate{R_z} & \qw  & \qw  & \qw  & \targ & \gate{R_y} & \gate{R_x} & \targ & \gate{R_y} & \gate{R_x} & \ctrl{1} & \qw  & \qw  & \qw  & \qw  & \qw  & \ctrl{1} & \gate{R_y} & \gate{R_z} & \qw  & \qw  & \qw  & \targ & \gate{R_y} & \gate{R_x} & \targ & \gate{R_y} & \gate{R_x} & \qw  \\
& \qw  & \gate{R_z} & \qw  & \qw  & \ctrl{-2} & \qw  & \qw  & \qw  & \qw  & \qw  & \ctrl{-1} & \gate{R_y} & \gate{R_z} & \targ & \gate{R_y} & \gate{R_x} & \targ & \gate{R_y} & \gate{R_x} & \targ & \gate{R_y} & \gate{R_x} & \targ & \gate{R_y} & \gate{R_x} & \ctrl{-1} & \gate{R_y} & \gate{R_z} & \ctrl{-1} & \gate{R_y} & \gate{R_z} & \qw  
}\]
\caption{Circuit for a randomly chosen isometry from one to three qubits. The
  two ancilla qubits are always initialized in the state $\ket{0}$,
  where an arbitrary state $\ket{\psi}$ is provided as an input on the
  least significant qubit. The output of the computation is read out
  from all three qubits at the end of the circuit. The circuit was
  produced by running
  $st$=\com{DecIsometry}$[$\com{PickRandomIsometry}$[2, 8]]$ in
  Mathematica and then calling \com{LatexQCircuit}$[st]$ to export the
  circuit to \LaTeX.  To save space we do not depict the angles here,
  but these are found by our code and can be output if desired.}
\label{fig:iso}
\end{figure*}
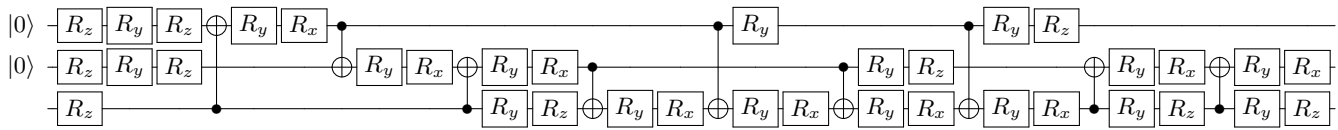

\subsection{Column-by-column decomposition} 

The column-by-column decomposition (method: \com{ColumnByColumnDec}) was introduced in~\cite{Iso} and achieves near optimal \cnot{} counts for generic isometries from $m$ to $n$ qubits. As the name suggests, the isometry is decomposed in a column-wise fashion (see~\cite{Iso} for the details). This decomposition achieves the lowest known \cnot{} counts for generic isometries with $1\leq m \leq n-2$. For isometries of a special form, it may also achieve lower \cnot{} counts for $m=0,n-1,n$ (after running the simplifications described in Section~\ref{sec:simplifying} and removing gates that implement the identity during the decomposition). In particular, it usually performs well for isometries with many zeros, since the number of gates to decompose the columns is reduced in such cases.

The column-by-column decomposition requires $2^{m+n}$ \cnots{} to
leading order for the decomposition of an isometry from $m$ to $n$
qubits. Its classical time complexity is  $\mathcal{O}(n2^{2m+n})$
(see Appendix~\ref{app:run_time}), which scales significantly better
in $m$ than the other decomposition methods for isometries. Note also
that it is straightforward to parallelize parts of the
column-by-column decomposition, which may help to lower the run time
significantly for practical implementations (but we have not done so
in the version 0.1 of our package).

\subsection{Quantum Shannon Decomposition (QSD)}

The Quantum Shannon Decomposition (method: \com{QSD}) was introduced
for unitaries in~\cite{2} and adapted to isometries in~\cite{Iso}. It
achieves lower \cnot{} counts than the column-by-column decomposition
for generic isometries from $m$ to $n$ qubits with $m=n-1$ or $m=n$,
whereas the QSD is not well adapted to the case $m \ll n$.  Its
classical time complexity is independent of $m$ and given by $\mathcal{O}(2^{3n})$ (see Appendix~\ref{app:classical_comp_QSD}).

\subsection{Knill's decomposition}

Knill's decomposition scheme (method: \com{KnillDec})
described in~\cite{Iso} and based on~\cite{Knill,3} expands an
isometry $V$ to a unitary $U$ maximizing the number of eigenvalues of
$U$ that are equal to one. The unitary $U$ can then be decomposed into
a circuit (described in~\cite{Knill,Iso}) that requires
$c \cdot (k+1) 2^{n} +k \mathcal{O}(n^2)$ \cnot{} gates for a unitary
on $n$ qubits with $k$ eigenvalues that are not equal to one, where
$c=23/24$ if $n$ is even and $c=115/96$ if $n$ is odd. For a generic
isometry $V$ from $m$ to $n$ qubits, the unitary extension, $U$, can
be chosen to have $2^m$ eigenvalues that are not equal to one and
hence requires $c \cdot (2^{n+m} + 2^{n}) +\mathcal{O}(2^{m+n/2})$
\cnot{} gates to leading order. However, for isometries of a special
form for which the unitary extensions has fewer than $2^m$ eigenvalues
that are not equal to one, Knill's decomposition may achieve lower
\cnot{} counts than the others (for an example, see the notebook
{\tt Examples.nb} that is provided together with the package).  The
classical time complexity of the decomposition is independent of
$m$ to leading order and given by $\mathcal{O}(2^{3n})$ (see
Appendix~\ref{app:classical_comp_Knill}).

\begin{figure*}[!t]
\[
\Qcircuit @C=0.2em @R=.2em {
\lstick{\left| 0 \right>}& \qw  & \gate{R_y\left(\textnormal{$\pi -2 \tan ^{-1}\left(\sqrt{2-\sqrt{3}}\right)$}\right)} & \gate{R_z(\textnormal{$\pi$})} & \targ & \gate{R_y\left(\textnormal{$\pi -2 \tan ^{-1}\left(\sqrt{2-\sqrt{3}}\right)$}\right)} & \qw  & \ctrl{1} & \meter &   &   \\
& \qw  & \gate{R_z(\textnormal{$\frac{3 \pi }{2}$})} & \qw  & \ctrl{-1} & \gate{R_y\left(\textnormal{$\frac{3 \pi }{2}$}\right)} & \gate{R_z(\textnormal{$\frac{3 \pi }{2}$})} & \targ & \gate{R_y(\textnormal{$\pi$})} & \gate{R_x(\textnormal{$\frac{3 \pi }{2}$})} & \qw  
}\]
\caption{Circuit for the amplitude damping channel given by Kraus
  operators $K_0(\gamma)=\{\{1, 0\}, \{0, \sqrt{1-\gamma}\}\}$ and
  $K_1(\gamma)=\{\{0, \sqrt{\gamma}\}, \{0, 0\}\}$ for
  $\gamma=1/3$. The ancilla qubit is always initialized in the state
  $\ket{0}$, where an arbitrary state $\ket{\psi}$ is provided as an
  input on the second qubit. The output of the computation is read out
  from the second qubit after tracing out  the first one (i.e., after
  measuring the first qubit and forgetting about the classical
  outcome). The circuit was produced by running
  $st$=\com{DecChannelInQCM}$[\{K_0(1/3),K_1(1/3)\}]$ in Mathematica
  and then \com{LatexQCircuit}$[st,AnglePrecision \to 1]$ to export
  the circuit to \LaTeX.}
\label{fig:channel}
\end{figure*}
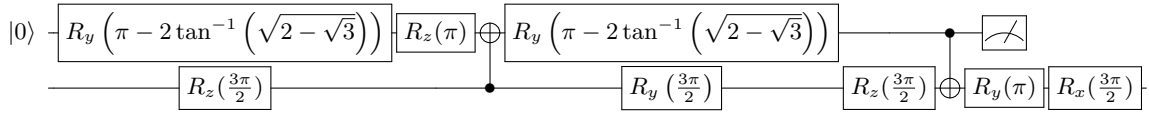

\subsection{Householder decomposition}\label{sec:house}
We also include a method for decomposing isometries using Householder reflections~\cite{house} (method: \com{DenseHouseholderDec}). A Householder reflection with respect to $\ket{v}$ is a unitary of the form $\id-2\proj{v}$. Generalizations of these can be used to construct unitaries that map any computational basis state to a particular state. Such unitaries can then be applied in sequence to construct an isometry, each mapping one basis state to the corresponding column.  Householder reflections are particularly useful for this because their successive action does not mess up previously created columns.  For a generic
isometry from $m$ to $n$ qubits, the number of \cnots{} required for this decomposition scales as $c \cdot (2^{n+m} + 2^{n}) +\mathcal{O}(2^{m+n/2})$, where $c=23/24$ if $n$ is even and $c=115/96$ if $n$ is odd~\cite{house}.  Note that this scaling is identical to that for Knill's decomposition; the advantage of using the Householder decomposition over Knill's is only for small $m$ and $n$. Our implementation of the Householder decomposition uses one ancilla qubit (which can start in any state and is returned to its initial state after the computation). The classical complexity of this decomposition is $\mathcal{O}(2^{m+n}(2^m+2^{n/2}))$~\cite{house}.

\subsection{State preparation} \label{sec:SP}

For the special case of state preparation on $n$ qubits (i.e., for an
isometry from $0$ to $n$ qubits), the best known decomposition scheme
is based on the Schmidt decomposition of the quantum state~\cite{3}
(method: \com{StatePreparation}). The scheme was slightly improved for
state preparation for an odd number of qubits in~\cite{Iso} leading to
a \cnot{} count of $23/24 \cdot 2^n$ for state preparation on $n$
qubits. This is lower than the number of \cnots{} required to
decompose an $n$-qubit state with uniformly controlled
gates~\cite{10}, which is $2^n$ to leading order\footnote{By default,
  the decomposition based on uniformly controlled gates~\cite{10} is
  used for the decomposition of the first column in the method
  \com{ColumnByColumnDec}. The option \com{FirstColumn} $\to$
  \com{StatePreparation} allows use of the scheme based on the Schmidt decomposition~\cite{3} for its decomposition.}. However, the classical time complexity is $\mathcal{O}(2^{\frac{3n}{2}})$ (see Appendix~\ref{app:classical_comp_SP} for the details), which is worse than the complexity  $\mathcal{O}(n2^{n})$  for state preparation using uniformly controlled gates.

\medskip

\noindent{\bf Remark} (States with low Schmidt rank).  The Schmidt
rank of a bipartite quantum state $\psi_{AB}$ is given by the minimal
number of Schmidt coefficients required for its Schmidt
decomposition. States with a small Schmidt rank correspond to weakly
entangled quantum systems, and occur naturally in the study of the
grounds states of certain types of Hamiltonian (see,
e.g.,~\cite{Hastings}). In the future, we plan to adapt the
decomposition for state preparation introduced in~\cite{3} to states
with low Schmidt ranks (where the splitting of the sub-systems has to
be specified). We expect this adaptation to lower the \cnot{} count
significantly for such states. Moreover, in the related task of
approximate state preparation, one could lower the gate count by
setting small Schmidt coefficients equal to zero.

\subsection{Sparse isometries}\label{sec:sparse}
We also have methods for implementing the decompositions of~\cite{house} that are designed for sparse isometries. These work using Householder reflections and can be seen as an adaptation of the method of Section~\ref{sec:house} that takes advantage of sparse structure. For sparse states of $n$ qubits with $\mathrm{nnz}$ non-zero entries, 
the decompositions result in $\mathcal{O}(n\cdot\mathrm{nnz})$ \cnots{} with a classical runtime of $\mathcal{O}({n \choose s}+n 2^{2s})$ with $s \in \mathbb{N}$ such that $\mathrm{nnz} \leq 2^s$  (method: \com{SparseStatePreparation}). The method is described in Section~2 and Section 5.2 in~\cite{house} and uses one ancilla.
For sparse isometries from $m\geq1$ to $n$ qubits (method: \com{SparseHouseholderDec}) it is more difficult to be precise about the number of required \cnots{}, since this number is affected by the structure of the non-zero elements in the isometry and not just the number of them.  The decomposition is based on Algorithm~3 in~\cite{house} (with minor modifications) and requires one ancilla.

\subsection{Choosing the optimal decomposition}
The method \com{DecIsometry}$[V]$ decomposes the isometry $V$ into a
sequence of single-qubit rotations and \cnot{} gates by running all
four decompositions (and in the case $m=0$ also the one for state
preparation) and in addition the decomposition for sparse $V$ if the number of zeros in $V$ is larger than $2^{m+\frac{n}{3}}$ , simplifying the gate sequences using the methods
described in Section~\ref{sec:simplifying}, and choosing the output
gate sequence that achieves the lowest \cnot{} count.  To decompose a
random isometry $V$ (of high dimensions), we recommend using
\com{DecIsometryGeneric}$[V]$, which chooses the decomposition method
that achieves the lowest \cnot{} count for a generic isometry with the
same dimensions as $V$ and hence has a shorter running time compared to \com{DecIsometry}$[V]$ (since it runs only one decomposition). \\

As sub-routines, we use optimal decompositions of two-qubit gates~\cite{optimal2qubit}  (method: \com{DecUnitary2Qubits}) and optimal state preparation on three qubits~\cite{OptimalSP} (method: \com{StatePrep3Qubits}).

\section{Compilation of quantum channels} \label{sec:channels}

In the following, we consider the implementation of quantum channels in the quantum circuit model (method: \com{DecChannelInQCM}). For the decomposition of channels, it is most convenient to work with the Kraus representation of the channel. Every quantum channel $\cE$ from $m$
to $n$ qubits with Kraus rank $K$ can be represented by Kraus operators $A_i$, which are complex matrices
 of dimension $2^n \times 2^m$ such that $\sum_{i=1}^{K} A_i^{\dagger}A_i=\id$ and
$\cE(\rho)=\sum_{i=1}^K A_i \rho A_i^{\dagger}$ for all
density operators $\rho$ of dimension $2^m$~\cite{choi}. 
To change the Kraus representation to a Choi state~\cite{choi} or vice versa, we provide the methods \com{KrausToChoi} and \com{ChoiToKraus}, respectively.

An arbitrary channel from $m$ to $n$ qubits (given as a list of Kraus
operators) can be provided as an input to \com{DecChannelInQCM}, which
returns a gate sequence (with some trace-out operations at the end) that implements the channel (see Figure~\ref{fig:channel} for an example).
The decomposition uses Stinespring's theorem~\cite{Stinespring} stating that a channel of Kraus rank $K=2^k$ can be represented by an isometry from $m$ to $n+k$ qubits. Then, the channel can be implemented by using one of the decomposition schemes for isometries and tracing out the ancillas at the end of the circuit. The \cnot{} count for a channel from $m$ to $n$ qubits of Kraus rank $2^k$ is therefore $2^{m+n+k}$ to leading order. This is nearly optimal for the decomposition of generic channels in the quantum circuit model~\cite{Iso}.

Note also that all channels from $m$ to $n$ qubits have a Kraus
representation with at most $2^{m+n}$ elements.  The command
\com{MinimizeKrausRank} is provided to do the reduction to the minimal
number of Kraus operators (which may be lower than $2^{m+n}$ for
channels of a special form).

\section{Compilation of POVM\lowercase{s}} \label{sec:POVM}

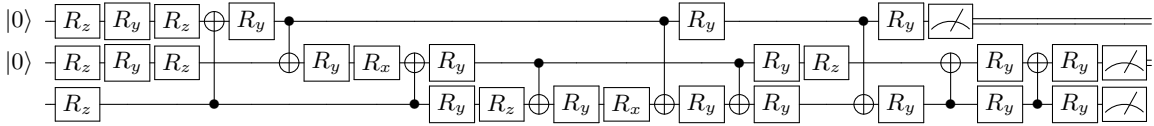
\begin{figure*}[!t]
\[
\Qcircuit @C=0.2em @R=.2em {
\lstick{\left| 0 \right>}& \qw  & \gate{R_z} & \gate{R_y} & \gate{R_z} & \targ & \gate{R_y} & \ctrl{1} & \qw  & \qw  & \qw  & \qw  & \qw  & \qw  & \qw  & \qw  & \ctrl{2} & \gate{R_y} & \qw  & \qw  & \qw  & \ctrl{2} & \gate{R_y} & \meter & \cw  & \cw  & \cw  & \cw  & \cw  \\
\lstick{\left| 0 \right>}& \qw  & \gate{R_z} & \gate{R_y} & \gate{R_z} & \qw  & \qw  & \targ & \gate{R_y} & \gate{R_x} & \targ & \gate{R_y} & \qw  & \ctrl{1} & \qw  & \qw  & \qw  & \qw  & \ctrl{1} & \gate{R_y} & \gate{R_z} & \qw  & \qw  & \targ & \gate{R_y} & \targ & \gate{R_y} & \meter & \cw  \\
& \qw  & \gate{R_z} & \qw  & \qw  & \ctrl{-2} & \qw  & \qw  & \qw  & \qw  & \ctrl{-1} & \gate{R_y} & \gate{R_z} & \targ & \gate{R_y} & \gate{R_x} & \targ & \gate{R_y} & \targ & \gate{R_y} & \qw  & \targ & \gate{R_y} & \ctrl{-1} & \gate{R_y} & \ctrl{-1} & \gate{R_y} & \meter &   
}\]
\caption{Circuit for a POVM on one qubit implementing the pretty good
  measurement for distinguishing the states $\phi_1=\ketbra{0}{0}$,
  $\phi_2=\ketbra{1}{1}$ and $\phi_3=\ketbra{+}{+}$ (the rotation
  angles are not depicted for simplicity). The POVM elements are given
  by $M_1= 1/4 \cdot \{\{1, 1\}, \{1, 1\}\}$, $M_2=1/8 \cdot \{\{3+2
  \sqrt{2},-1\},\{-1,3-2 \sqrt{2}\}\}$ and $M_3=1/8 \cdot \{\{3 - 2
  \sqrt{2}, -1\}, \{-1, 3 + 2  \sqrt{2}\}\}$. The probability to
  measure $\ket{i}$ (with $i \in \{ 1,2,3\}$) on the two ancilla
  qubits for a given state $\rho$ on the third qubit is given by $\tr(M_i \rho)$. The circuit was produced by running
  $st$=\com{DecPOVMInQCM}$[\{M_1,M_2,M_3\}]$ in Mathematica and then
  \com{LatexQCircuit}$[st]$ to export the circuit to \LaTeX. Again, the rotation angles are not depicted for simplicity.}
\label{fig:POVM}
\end{figure*}

Positive-operator valued measures (POVMs) describe the most general
measurements that can be performed on quantum systems. In the
following, we consider the implementation of POVMs on $m$ qubits in
the quantum circuit model (method: \com{DecPOVMInQCM}). Every POVM
$\mathcal{M}$ on $m$ qubits with $L$ possible measurement outcomes can
be represented by $L$ operators $0\leq E_{i}\leq \id$ satisfying
$\sum_{i=0}^{L-1} E_i =\id$. The probability of getting the
$i^{\text{th}}$ outcome when performing the POVM on an $m$ qubit state
$\rho$ is then given by $\tr[E_i\rho]$.

To demonstrate the use of \com{DecPOVMInQCM}, we consider state
discrimination (see for example~\cite{BarnettCroke} for a review). Suppose a state is chosen from a known set of (not necessarily
pure) density operators $\{\phi_i\}_i$, where
$\phi_i$ is chosen with probability
$p_i$.  The goal is to correctly guess which state was chosen, by
performing a measurement on the given state. In general, the optimal
strategy for such a task is to perform a (non-projective) POVM. Since
it is (in general) difficult to find the optimal POVM, a ``pretty
good'' choice was introduced in~\cite{PGM1,PGM2}.  Using
\com{DecPOVMInQCM} we can find a quantum circuit to implement these
POVM elements. Running this circuit on some quantum hardware would
then give us a (classical) output that corresponds to a pretty good
guess of which state was given to us. As a concrete example, we take 
$\phi_1=\ketbra{0}{0}$,
$\phi_2=\ketbra{1}{1}$ and
$\phi_3=\ketbra{+}{+}$, where
$\ket{+}:=1/\sqrt{2}\,(\ket{0}+\ket{1})$, assuming that each state is
chosen with the same probability
$p_i=1/3$. The pretty good measurement has POVM elements $M_i=p_i
\phi^{-1/2} \phi_i \phi^{-1/2}$, where $\phi=\sum_i p_i
\phi_i$. Using \com{DecPOVMInQCM} gives a quantum circuit with two
ancilla qubits (see Figure~\ref{fig:POVM}). The outcomes can be
interpreted as follows:
$(x,y)=(0,0)$ corresponds to a guess that the chosen state was
$\phi_1$, $(x,y)=(0,1)$ corresponds to a guess of
$\phi_2$ and $(x,y)=(1,0)$ corresponds to a guess of
$\phi_3$. Note that $(x,y)=(1,1)$ has probability zero.

\section{Compilation of instruments}\label{sec:inst}
Instruments are generalizations of both channels and POVMs. Again we
consider an implementation in the quantum circuit model (method:
\com{DecInstrumentInQCM}). To specify an instrument, we give the Kraus
operators of the subnormalized channels comprising the instrument.
Consider an instrument from $m$ to $n$ qubits with $L$ outcomes. This
can be specified using Kraus operators of dimension $2^n\times 2^m$.
If the $j^{\text{th}}$ subnormalized channel has Kraus representation
$\{A^j_i\}_{i=1}^{K_j}$, then to be a valid instrument requires
$\sum_{j=1}^L\sum_{i=1}^{K_j}(A^j_i)^\dagger A^j_i=\id$.  Starting
from state $\rho$, the outcome $j$ occurs with probability
$\tr(\cE^j(\rho))$ and the post-measurement state is
$\cE^j(\rho)/\tr(\cE^j(\rho))$, where
$\cE^j(\rho)=\sum_{i=1}^{K_j}A^j_i\rho(A^j_i)^\dagger$. In general,
the circuit output by \com{DecInstrumentInQCM} will involve some final
measurements and trace out operations, as well as some output qubits.

A circuit for an instrument can be converted to one for the
corresponding POVM by tracing out all of the output qubits, although
the circuit formed by this procedure may be longer than that formed by
first calculating the POVM and decomposing it using
\com{DecPOVMinQCM}.  Furthermore, if one takes the circuit output by
\com{DecPOVMinQCM} for the POVM $\{M_1,M_2,\ldots\}$ and removes all
trace out operations at the end, then one gets a circuit for the
instrument $\{\{\sqrt{M_1}\},\{\sqrt{M_2}\},\ldots\}$, i.e., where
each subnormalized channel has just one Kraus operator equal to the
square root of the corresponding POVM element.

\section{Simplifications of gate sequences} \label{sec:simplifying}

Given a sequence of single-qubit rotations and \cnot{} gates, it may
be possible to find a shorter gate sequence that implements the same
operation. We provide the method \com{SimplifyGateList}, which uses
some simple rules to reduce the gate count. The number of single-qubit
gates is reduced by merging single-qubit rotations in cases where more
than two occur consecutively on the same qubit. The merged single-qubit unitary can then be decomposed using the following well-known Lemma~\cite{Buch}.

\begin{lem}[ZYZ decomposition]\label{ZYZ} 
  For every unitary operation $U$ acting on a single qubit, there
  exist real numbers $\alpha,\beta,\gamma$ and $ \delta$ such that
\begin{equation} \label{eq7}
	U=e^{\I\alpha}R_{z}(\beta)R_{y}(\gamma)R_{z}(\delta).
\end{equation}
\end{lem}

By symmetry, Lemma~\ref{ZYZ} holds for any two orthogonal rotation
axes. We decompose the merged unitary using the ZYZ decomposition if the previous \cnot{} gate controls on the considered qubit, and we decompose it using the XYX decomposition otherwise. Since $R_z$ gates commute with the control of \cnot{} gates and $R_x$ gates with the target, one of the rotation gates can be commuted to the left of the \cnot{}, as summarized in the following circuit equivalence.

\[
\Qcircuit @C=0.8em @R=.46em {
 & \ctrl{2} &\gate{R_z}&\gate{R_y}&\gate{R_z}&\qw&&& &\gate{R_z}& \ctrl{2} &\gate{R_y}&\gate{R_z}&\qw   \\
& &&&&&&	 =&& \\
 &\targ &\gate{R_x}&\gate{R_y}&\gate{R_x}&\qw& &&  &\gate{R_x}&\targ &\gate{R_y}&\gate{R_x}&\qw	 \\
}
\]

We do this procedure starting from the end of the circuit and we also
cancel \cnot{} gates where we have two in a row (or with commuting
\cnot{} gates in-between them) with the same control and target. The
resulting circuit contains at most four single-qubit rotations after
each \cnot{} gate. Note that to do the simplifications we have to
traverse only once through the circuit, hence the classical run time of
this procedure is linear in the number of gates of the circuit.

For example, the following circuit (for arbitrary rotation angles)

\[
\Qcircuit @C=0.8em @R=.46em {
& \qw  & \ctrl{1} & \targ & \gate{R_z} & \ctrl{1} & \gate{R_z} & \qw  & \ctrl{1} & \qw  & \qw  \\
& \qw  & \targ & \ctrl{-1} & \gate{R_x} & \targ & \targ & \gate{R_x} & \targ & \qw  & \qw  \\
& \qw  & \gate{R_z} & \qw  & \qw  & \qw  & \ctrl{-1} & \gate{R_y} & \gate{R_x} & \gate{R_z} & \qw  
}
\]

gets simplified to the following (by only traversing the circuit once).

\[
\Qcircuit @C=0.8em @R=.46em {
& \qw  & \ctrl{1} & \targ & \gate{R_z} & \qw  & \qw  & \qw  & \qw  \\
& \qw  & \targ & \ctrl{-1} & \gate{R_x} & \targ & \qw  & \qw  & \qw  \\
& \qw  & \gate{R_z} & \qw  & \qw  & \ctrl{-1} & \gate{R_y} & \gate{R_z} & \qw  
}
\]

\section{Exporting circuits to {{O\lowercase{pen}QASM}}} \label{sec:QASM}

We also provide a python script to translate the gate sequences produced by Mathematica to the QASM language~\cite{QASM}. The gate sequence can then be imported into for example the IBM library \textit{Qiskit} and further optimized or also directly sent to quantum hardware for evaluation. The script is based on \textit{ProjectQ}~\cite{ProjectQ1,ProjectQ2} and is simple to use (see our documentation for more details).

\section{Future work}
 In a future version we plan to implement code for the decompositions
of quantum channels and POVMs in more general models than the quantum
circuit model that allow for measurements in between the gate sequence
and to classical control on the measurement
results~\cite{channel}. This will significantly reduce the \cnot{}
count for channels and POVMs. 


A remaining open question is how to use these decomposition as a
starting point for circuit optimization. As a straightforward
application, one could do peephole optimization, by taking a large
circuit and extracting parts of it that act on, e.g., three qubits,
and resynthesize the unitary corresponding to the circuit. If this
leads to a shorter circuit, this could then replace the original.
Alternatively, one could build up sets of increasingly complicated
templates~\cite{maslov}, i.e., circuits that implement the identity
operator, using our universal decomposition schemes. Indeed, choosing
a (parametrized) circuit, writing it as a unitary and synthesizing a
new circuit for it, directly leads to a (parametrized) template. These
templates could then be used to simplify parts of larger circuits.

\section{Acknowledgements}
We thank Thomas H\"aner, Dmitri Maslov and Joseph
M.\ Renes for helpful discussions.  We are grateful to the Department
of Mathematics, University of York for part-funding some summer
projects that enabled this work.  RI acknowledges support from the
Swiss National Science Foundation through SNSF project No.\
200020-165843 and through the National Centre of Competence in
Research \textit{Quantum Science and Technology} (QSIT).

\appendix

\section{Transforming our gate library to one that is well adapted for trapped ions} \label{app:tranform_gate_library}
The common universal gate set used for trapped ions consists of single-qubit gates
$\rgate(\theta, \phi)$ and the M{\o}lmer-S{\o}rensen gate $\xx(\phi)$ (see for example~\cite{home}) defined as follows.
\begin{align*}
    \rgate(\theta, \phi) = \begin{bmatrix}
        \cos(\theta/2)& -i \exp(-i \phi) \sin(\theta/2)\\
        -i\exp(i \phi) \sin(\theta/2) & \cos(\theta/2)\\
    \end{bmatrix},\\
    \xx(\phi) = \begin{bmatrix}
        \cos(\phi)&0&0&-i\sin(\phi)\\
        0&\cos(\phi)&-i\sin(\phi)&0\\
        0&-i\sin(\phi)&\cos(\phi)&0\\
        -i\sin(\phi)&0&0&\cos(\phi)\\
    \end{bmatrix}.\\
\end{align*}
In particular, we have $R_x(\theta)=\rgate(-\theta,0)$ and $R_y(\theta)=\rgate(-\theta,\pi/2)$.
Having a circuit consisting of \cnot{} gates, one can use the
following identity to replace each \cnot{} gates with a single  \xx
gate and single qubit gates~\cite{Maslov_trapped_ions}. Note that this transformation does not increase the two-qubit gate count.
\[
\Qcircuit @C=0.8em @R=.46em {
    & \ctrl{2} &\qw&&&\gate{R_y(-\frac\pi2)}&\multigate{2}{\text{XX}(\frac\pi4)}&\gate{R_x(\frac\pi2)}&
    \gate{R_y(\frac\pi2)}&\qw\\
    &&&=&&& \\
    &\targ &\qw&&&\qw &\ghost{\text{XX}(\frac\pi4)}&\gate{R_x(\frac\pi2)}&\qw&\qw \\
}
\]
In the following, we show how to merge the single-qubit gates in a circuit containing  \xx and single-qubit rotations, such that the resulting circuit  contains at most one $\rgate(\theta, \phi)$ gate after each \xx gate, and additionally a possible $R_x$ gate on each of the qubits at the beginning of the circuit. To do so, we use  that the \xx gate commutes with the $R_x$ gate (on both qubits)~\cite{Maslov_trapped_ions} together with the following decomposition.

\begin{lem}[\rgate-$R_x$ decomposition]
\label{Rxr}
Given a $2\times2$ unitary matrix, $U$, there exist reals $\alpha, \theta, \phi
$ and $\delta$ such that
\begin{equation} 
    U=e^{\I\alpha}\rgate(\theta, \phi)R_{x}(\delta).
\end{equation}
\end{lem}
\begin{proof}
    From (the generalized) Lemma~\ref{ZYZ}, it follows that there exist reals $\alpha, \beta, \gamma, \tilde\delta$,
    such that
    \begin{equation} 
        U=e^{\I\alpha}R_{x}(\beta)R_y(\gamma)R_x(\tilde\delta).
    \end{equation}
    The circuit equivalence (12) in~\cite{Maslov_trapped_ions} implies that there exists reals
    $\theta, \phi$ such that $\rgate(\theta,\phi) R_x(-\beta)R_y(-\gamma)R_x(-\beta) =
    \id$, or, equivalently, $R_x(\beta) R_y(\gamma)=\rgate(\theta,\phi) R_x(-\beta)$.
    It follows that
    \begin{align}
        U&=e^{\I\alpha}R_x(\beta)R_y(\gamma)R_x(\tilde\delta)\\
         &=e^{\I\alpha}\rgate(\theta,\phi) R_x(-\beta) R_x(\tilde\delta) \\
         &=e^{\I\alpha}\rgate(\theta,\phi) R_x(\delta) \, ,
    \end{align}
    where $\delta := \tilde\delta - \beta$.
  \end{proof}
  
This leads to the circuit equivalence
\[
\Qcircuit @C=0.8em @R=.46em {
    &\multigate{2}{\text{XX}(\phi)}&\gate{U}&\qw&&&&&\gate{R_x}&\multigate{2}{\text{XX}(\phi)}&\gate{\rgate}&\qw\\
    &&&&&=&&& \\
    &\ghost{\text{XX}(\phi)} &\gate{U} &\qw&&&&&\gate{R_x}&\ghost{\text{XX}(\phi)}&\gate{\rgate}&\qw \\
}
\]
which we can apply recursively starting at the end of the circuit as
follows.  We first merge all the single-qubit rotations into
single-qubit unitaries before applying the above circuit equivalence at
the last \xx gate appearing in the circuit and merge the $R_x$ gates
into the proceeding single-qubit unitary.  We then apply the circuit
equivalence to the second to last \xx gate in the circuit, and so
on. The single-qubit unitaries that remain at the end can be
written in terms of $\rgate$-gates.

We provide the commands
\com{CNOTRotationsToXXRGates} and \com{XXRGatesToCNOTRotations} to
convert between circuits that use single qubit
rotations and \cnots{} and those using \xx and
$\rgate$-gates.

\section{Classical time complexity for the decomposition of isometries} \label{app:run_time}
 In this section we give some details about how to find the classical
 time complexity for the different decomposition schemes for
 isometries. Note that these complexities refer to numerical cases
 (not analytic calculations).
 
\subsection{Classical complexity for the column-by-column decomposition} 
To leading order, we only have to consider the decomposition and
simulation (i.e., the application to quantum states) of the uniformly
controlled gates. The decomposition scheme for a uniformly controlled
single-qubit gate with $k$ controls introduced in~\cite{10} has time
complexity $\mathcal{O}(k2^k)$, and computing the updated state after
its application to an $n$-qubit state has time complexity
$\mathcal{O}(2^{n})$ (one has to update all the entries of the state
vector in general). Hence, to update all $2^m$ columns of an isometry
from $m$ to $n$ qubits has time complexity
$\mathcal{O}(2^{m+n})$.\footnote{Decomposing the $l^{\text{th}}$ column with the column-by-colum decomposition, the columns $1,\dots,l-1$ are in the states $e^{i \phi_0}\ket{0},\dots, e^{i \phi_{l-2}} \ket{l-2}$ for some real phases $\phi_0,\dots,\phi_{l-2}$, and hence the application of uniformly controlled gates on these columns has constant time complexity. We ignore this here, since it does not change the overall time complexity of the column-by-column decomposition.}  Note that it is
straightforward to parallelize the application of a uniformly
controlled gate to the different columns of an isometry (and also the application to a single-column), which can speed up practical
implementations significantly (but we do not take this into account
here for the complexity measure). The complexity of decomposing one
column is
\begin{align*}
\sum_{k=0}^{n-1} \mathcal{O}(k2^k+2^{m+n})=\mathcal{O}(n2^{m+n})
\end{align*}
Hence, the complexity to decompose all of the $2^m$ columns is $\mathcal{O}(n2^{2m+n})$. 

\subsection{Classical complexity for the Quantum Shannon Decomposition} \label{app:classical_comp_QSD}
The Quantum Shannon Decomposition of an isometry from $m$ to $n$ qubits is based on the Cosine-Sine-Decomposition of an unitary expansion of the isometry~\cite{2,Iso}. Since the unitary expansion is a matrix of dimension $2^n \times 2^n$, the time complexity to perform the Cosine-Sine-Decomposition of it is  $\mathcal{O}(2^{3n})$~\cite{CSD}. 

\subsection{Classical Complexity for Knill's decomposition} \label{app:classical_comp_Knill}

Knill's decomposition of an isometry from $m$ to $n$ qubits requires running several subroutines from linear algebra to find the unitary $U$ from Lemma~\ref{lemma_existence_U} and its eigenvalue decomposition, which is required for the decomposition (see~\cite{Knill,Iso} and Appendix~\ref{app:Knill} for the details). The most time consuming operations are:
\begin{itemize}
\item Finding an orthonormal basis of the null space of $V^\dagger-I_{2^m,2^n}$ of dimension $2^m\times 2^n$ in the proof of  Lemma~\ref{lemma_existence_U},
\item Multiplying the matrices $W$ and $W''$ of dimension $(2^n-2^m)\times(2^n-2^m)$ in the proof of  Lemma~\ref{lem:Kn1}, 
\item Finding the eigenvalues and eigenvectors of a $2^n \times 2^n$ unitary $U$ in Lemma~\ref{lemma_Knill_dec}.
\end{itemize}
All of these operations can be implemented with time complexity $\mathcal{O}(2^{3n})$.\\

The time complexity of the remaining part of Knill's decomposition is dominated by the decomposition of the state preparation operations, which are denoted by $V_i$ in Lemma~\ref{lemma_Knill_dec}. By Appendix~\ref{app:classical_comp_SP}, state preparation (using the decomposition scheme introduced in~\cite{3}) has time complexity $\mathcal{O}(2^{3n/2})$. For Knill's decomposition, we have to perform this decomposition $2^{m+1}$ times, and hence the ``decomposition-phase'' of Knill's scheme has time complexity $\mathcal{O}(2^{m+3n/2})$.\\

We conclude that the whole decomposition has time complexity $\mathcal{O}(2^{3n}+2^{m+3n/2})=\mathcal{O}(2^{3n})$ (since $m\leq n$).

\subsection{Classical complexity for state preparation} \label{app:classical_comp_SP}

The method introduce in~\cite{3} requires calculating the Schmidt
decomposition of the given state on $n$ qubits as well as decomposing
two unitaries on each half of the qubits. To calculate the Schmidt
decomposition, one performs the singular value decomposition on a
matrix of dimension
$2^{ \lfloor n/2 \rfloor} \times 2^{\lceil n/2 \rceil}$, which has
complexity $\mathcal{O}(2^{3n/2})$~\cite{GolubvanLoan}. Decomposing
the unitaries can also be done with complexity $\mathcal{O}(2^{3n/2})$
(see Appendix~\ref{app:classical_comp_QSD}). We conclude that state
preparation as done in~\cite{3} has classical time complexity
$\mathcal{O}(2^{3n/2})$.

\section{Theoretical details required for the implementation of Knill's decomposition} \label{app:Knill}

In this appendix, we give some details about Knill's decomposition
introduced in~\cite{Knill} that are required for its implementation.
To help keep track of dimensions, throughout this section we will use
$\{\ket{i}_D\}_{i=0}^{D-1}$ to denote an orthonormal basis for
$\mathbb{C}^D$ for any $D\in\mathbb{N}$.

\begin{lem} \label{lemma_Knill_dec}
  Let $U$ be an $\n\times\n$ unitary matrix with eigendecomposition
  $U=\id+\sum_{i=0}^{t-1}(e^{i\theta_i}-1)\proj{v_i}$, with $\{\ket{v_i}\}$
  orthonormal, $e^{i\theta_i}\neq1$ and $t\geq1$ (i.e., $U$ has $t$ eigenvalues that
  differ from $1$). Then $U=\prod_{i=0}^{t-1}V_iP_iV_i^\dagger$, where
  $V_i$ is any unitary that takes $\ket{0}$ to $\ket{v_i}$ and
  $P_i=\id+(e^{i\theta_i}-1)\proj{0}$.
\end{lem}
\begin{proof}
  Write $V_i=\ketbra{v_i}{0}+R_i$, where
  $R_i=\sum_{j=1}^{t-1}\ketbra{r^i_j}{j}$ so that $R_i\ket{0}=0$.  In
  order that $V_i$ is unitary, we require
  $R_iR_i^\dagger=\id-\proj{v_i}$.  Then,
\begin{eqnarray*}
V_iP_iV_i^\dagger&=&(\ketbra{v_i}{0}+R_i)(\id+(e^{i\theta_i}-1)\proj{0})(\ketbra{0}{v_i}+R^\dagger_i)\\
&=&e^{i\theta_i}\proj{v_i}+R_iR_i^\dagger\\
&=&\id+(e^{i\theta_i}-1)\proj{v_i}\, .
\end{eqnarray*}
The product is hence $U$.
\end{proof}

We now show (along the lines of~\cite{Knill}) that any $\n\times\m$
isometry can be extended to a unitary with at most $\m$ eigenvalues
that differ from $1$.

\begin{lem}\label{lem:Kn1}
  Let $X$ and $Y$ be $\n\times\m$ matrices such that $X^\dagger
  X=Y^\dagger Y$. Then there exists an $\n\times\n$ unitary $U$ such
  that $UX=Y$.
\end{lem}
\begin{proof}
  Let $\m\leq\n$, and $Y=W\Sigma V$ be the SVD of $Y$, with
  $\Sigma=\sum_{i=1}^\m\sigma_i \ket{i}_{\n}\bra{i}_\m$ where
  $\{\sigma_i\}$ are non-negative real numbers. Note that $W$ is
  $\n\times\n$, $V$ is $\m\times\m$ and $\Sigma$ is $\n\times\m$.  

  We have $Y^\dagger Y=V^\dagger \Sigma^\dagger\Sigma V$.  Thus
  $VX^\dagger XV^\dagger=\Sigma^\dagger
  \Sigma=\sum_i\sigma_i^2\proj{i}_\m$. Let
  $XV^\dagger=\sum_i\ketbra{v_i}{i}_\m$, for some (non-normalized)
  vectors $\{v_i\}_{i=1}^\m$; $v_i\in\mathbb{C}^\n$. Then
  $VX^\dagger
  XV^\dagger=\sum_{i,j}\braket{v_i}{v_j}\ket{i}_\m\!\bra{j}_\m$, hence
  $\braket{v_i}{v_j}=\sigma_i^2\delta_{ij}$.  We hence define the
  $\n\times\n$ matrix
  $W'=\sum_{i:\sigma_i\neq0}\sigma_i^{-1}\ket{i}_\n\!\bra{v_i}$, so
  that $W'XV^\dagger=\Sigma$. Note that
  $W'(W')^\dagger=\sum_{i:\sigma_i\neq0}\ket{i}_\n\!\bra{i}_\n$, and
  that $W'$ can be extended to a unitary $W''$ without affecting its
  action on $XV^\dagger$. Then, if we take $U=WW''$, we have
  $UX=WW''X=W\Sigma V=Y$.  The case $\m>\n$ can be treated similarly.
\end{proof}

\begin{lem} \label{lemma_existence_U}
  Let $\n$ and $\m\leq\n$ be positive integers and $V$ be an $\n\times\m$ matrix
  satisfying $V^\dagger V=\id_\m$, i.e., $V$ is an isometry. There exists an $\n\times\n$ unitary matrix $U$
  such that $U\ket{i}_\n=V\ket{i}_\m$ for $i\in\{0,\ldots,\m-1\}$, and
  which has at least $\n-\m$ eigenvalues equal to 1.
\end{lem}
\begin{proof}
  First note that $V$ can be written in terms of its columns
  $\ket{v_i}\in\mathbb{C}^\n$ via
  $V=\sum_{i=0}^{\m-1}\ketbra{v_i}{i}_\m$. We can then find $\n-\m$
  further vectors $\ket{v_i}\in\mathbb{C}^\n$ numbered from $i=\m$ to
  $\n-1$ such that $\{\ket{v_i}\}_{i=0}^{\n-1}$ is an orthonormal basis
  for $\mathbb{C}^\n$.  The matrix
  $\sum_{i=0}^{\n-1}\ketbra{v_i}{i}_\n=:\tilde{U}$ is then a unitary
  satisfying $\tilde{U}\ket{i}_\n=V\ket{i}_\m$ for
  $i\in\{0,\ldots,\m-1\}$.  We need to show that it is always possible
  to choose the set $\{\ket{v_i}\}_{i=\m}^{\n-1}$ such that $U$ has at
  least $\n-\m$ eigenvalues equal to 1.  Note that this is equivalent to
  $U^\dagger$ having at least $\n-\m$ eigenvalues equal to 1.

  Let us write
  $U^\dagger=\left(\begin{array}{c}V^\dagger\\W\end{array}\right)$,
  where $W$ is $(\n-\m)\times\n$ so that for any $\n\times K$ matrix
  $X$ for some positive integer $K$,
  $U^\dagger X=\left(\begin{array}{c}V^\dagger
      X\\WX\end{array}\right)$.  Note that, by unitarity,
  $VV^\dagger+W^\dagger W=\id_\n$.

  $V^\dagger-I_{\m,\n}$ has dimension $\m\times\n$ and hence its
  nullspace dimension $q$ is at least $\n-\m$ (here $I_{\m,\n}$ denotes
  the $\m\times\n$ matrix
  $I_{\m,\n}=\sum_{i=0}^{\m-1}\ket{i}_\m\!\bra{i}_\n$). Let us take
  $\ket{f_i}\in\mathbb{C}^\n$ to be an orthonormal basis for this
  nullspace for $i\in\{0,\ldots,q-1\}$ so that
  $V^\dagger\ket{f_i}=\ket{f_i}$ for $i\in\{0,\ldots,q-1\}$.

  Consider now the $\n\times(\n-\m)$ matrix
  $X=\sum_{i=0}^{q-1}\ketbra{f_i}{i}_{\n-\m}$.  We can rewrite $X$ in
  terms of its rows as $X=\sum_{i=1}^\n\ket{i}_\n\!\bra{x_i}$, where
  $\ket{x_i}\in\mathbb{C}^{\n-\m}$ and divide this into $X_1$ and
  $X_2$, where $X_1$ comprises the first $\m$ rows, and $X_2$ the
  remaining $\n-\m$ rows (e.g.,
  $X_1=\sum_{i=0}^{\m-1}\ket{i}_\n\!\bra{x_i}\,$).  By construction,
  $V^\dagger X=I_{\m,\n} X=X_1$ and hence
  $X_1^\dagger X_1=X^\dagger VV^\dagger X$.  Furthermore,
  $\id_{\n-\m}=X^\dagger X=X_1^\dagger X_1+X_2^\dagger X_2$ and hence
  $X^\dagger W^\dagger WX=X^\dagger(\id-VV^\dagger)X=X^\dagger
  X-X_1^\dagger X_1=X_2^\dagger X_2$.

  Since there is unitary freedom in $W$, it follows from
  Lemma~\ref{lem:Kn1} that it can be chosen such that $WX=X_2$.  With
  this choice, $U^\dagger X=X$, and hence the $\n-\m$ columns of $X$ are
  eigenvectors of $U^\dagger$ with eigenvalue 1.
\end{proof}

%

\end{document}